\documentclass[
submission
]{dmtcs-episciences}

\usepackage{mathptmx}      
\usepackage{color}
\usepackage{amsmath, amssymb, amsfonts, amsfonts,latexsym}
\usepackage{relsize}
\usepackage{algorithmic,algorithm}
\usepackage{subfigure}

\usepackage[numbers]{natbib}

\usepackage{mathtools}

\makeatletter
\providecommand*{\cupdot}{%
  \mathbin{%
    \mathpalette\@cupdot{}%
  }%
}
\newcommand*{\@cupdot}[2]{%
  \ooalign{%
    $\m@th#1\cup$\cr
    \hidewidth$\m@th#1\cdot$\hidewidth
  }%
}
\makeatother
\makeatletter
\newcommand{\raisemath}[1]{\mathpalette{\raisem@th{#1}}}
\newcommand{\raisem@th}[3]{\raisebox{#1}{$#2#3$}}
\providecommand*{\sqplus}{%
  \mathbin{%
    \mathpalette\@sqplus{}%
  }%
}
\newcommand*{\@sqplus}[2]{%
  \ooalign{%
    $\m@th#1\sqcup$\cr
    \hidewidth$\m@th#1\raisemath{1.4pt}{\scriptscriptstyle{+}}$\hidewidth
  }%
}
\makeatother

\newcommand{\DELTA}{\bigtriangleup}
\newcommand{\M}{\ensuremath{\operatorname{\mathbb{M}}}}
\newcommand{\T}{\ensuremath{\operatorname{\mathcal{T}}}}
\renewcommand{\P}{\ensuremath{\operatorname{\mathbb{P}}}}
\newcommand{\Pmax}{\ensuremath{\operatorname{\mathbb{P_{\max}}}}}

\newcommand{\Pmaxi}[1]{\ensuremath{\operatorname{\mathbb{P}^{#1}_{\max}}}}

\newcommand{\MD}{\ensuremath{\operatorname{MD}}}
\newcommand{\MDT}{\ensuremath{\operatorname{MDT}}}
\newcommand{\MDP}{\ensuremath{\operatorname{MDP}}}
\newcommand{\opt}{\ensuremath{\operatorname{opt}}}
\newcommand{\Go}{\ensuremath{G_{\opt}}}
\newcommand{\Fo}{\ensuremath{F_{\opt}}}
\newcommand{\out}[1]{out$_{#1}$}
\newcommand{\lca}{\ensuremath{\operatorname{lca}}}
\newcommand{\merge}{\ensuremath{\sqplus}}
\newcommand{\Merge}{\ensuremath{\displaystyle{\sqplus}}}

\newcommand{\eqcl}{\mathrel{\mathsmaller{\mathsmaller{^{\boldsymbol{\sqsubseteq}}}}}}
\newcommand{\mb}{\ensuremath{\mathbb}}
\newcommand{\mc}{\ensuremath{\mathcal}}

 \newtheorem{theorem}{Theorem}[section]
 \newtheorem{lemma}[theorem]{Lemma}
 \newtheorem{corollary}[theorem]{Corollary}
 \newtheorem{proposition}[theorem]{Proposition}
 \newtheorem{problem}[theorem]{Problem}
 \newtheorem{definition}{Definition}
 \newtheorem{remark}[theorem]{Remark}

\author{Marc Hellmuth\affiliationmark{1,2}\thanks{Email: \texttt{mhellmuth@mailbox.org} (Corresponding Author)}   
  \and Adrian Fritz\affiliationmark{2}
  \and Nicolas Wieseke\affiliationmark{3}\thanks{
    Supported  
    by the German Research Foundation (DFG) (Proj.\ No.\
    MI439/14-1}
   \and Peter F. Stadler\affiliationmark{4-8}
}
\title[Techniques for the  Cograph Editing Problem]{Techniques for the  Cograph Editing Problem: \\ Module Merge is equivalent to Editing $P_4$'s}
\affiliation{
  Department of Mathematics and Computer Science, University of Greifswald, Greifswald, Germany\\
  Center for Bioinformatics, Saarland University, Saarbr{\"u}cken, Germany\\
  Parallel Computing and Complex Systems Group, Department of Computer Science, Leipzig
  University, Germany \\
  Bioinformatics Group, Dept.\ of Computer Science \& Interdisciplinary
  Center of Bioinformatics, Leipzig University\\
  Max Planck Institute for Mathematics in the Sciences, Leipzig \\
  Institute for Theoretical Chemistry, University of Vienna, Austria\\
  Center for non-coding RNA in Technology and Health, Copenhagen Univ., 
  Denmark\\
  The Santa Fe Institute, Santa Fe, USA}
\received{1998-10-14}
\revised{2002-07-19, 2014-02-05, 2015-09-09}
\accepted{2015-09-09}
\begin{document}
\publicationdetails{VOL}{2015}{ISS}{NUM}{SUBM}
\maketitle
\begin{abstract}
  Cographs are graphs in which no four vertices induce a simple connected
  path $P_4$. Cograph editing is to find for a given graph $G = (V,E)$ a
  set of at most $k$ edge additions and deletions that transform $G$ into a
  cograph. This combinatorial optimization problem is NP-hard. It has,
  recently found applications in the context of phylogenetics, hence good
  heuristics are of practical importance.  \newline 
  It is well-known that the cograph editing problem can be solved
  independently on the so-called strong prime modules of the modular
  decomposition of $G$. We show here that editing the induced $P_4$'s of a
  given graph is equivalent to resolving strong prime modules by means of a
  newly defined merge operation $\merge$ on the submodules. This observation
  leads to a new exact algorithm for the cograph editing problem that can
  be used as a starting point for the construction of novel heuristics.
\end{abstract}

\sloppy

\section{Introduction}

Cographs are among the best-studied graph classes. In particular the fact
that many problems that are NP-complete for arbitrary graphs become
polynomial-time solvable on cographs \cite{Corneil:85, BLS:99,Gao:13} makes
them an attractive starting point for constructing heuristics.  As noted
already in \cite{Corneil:81}, the input for several combinatorial
optimization problems, such as exam scheduling or several variants of
clustering problems, is naturally expected to have few induced
$P_4$s. Since graphs without an induced $P_4$ are exactly the cographs,
identifying the closest cograph and solving the problem at hand for the
modified input becomes a viable strategy.

It was recently shown that \emph{orthology}, a key concept in evolutionary
biology in phylogenetics, is intimately tied to cographs. Two genes in a
pair of related species are said to be orthologous if their last common
ancestor was a speciation event.  The orthology relation on a set of genes
forms a cograph \cite{HHH+13}.  This relation can be estimated directly
from biological sequence data, albeit in a necessarily noisy
form. Correcting such an initial estimate to the nearest cograph, i.e.,
cograph-editing, thus has recently become an computational problem of
considerable practical interest in computational biology
\cite{HWL+15}. However, the (decision version of the) problem to edit a
given graph into a cograph is NP-complete \cite{Liu:11,Liu:12}.  We showed
in \cite{HWL+15} that the cograph-editing problem is amenable to
formulations as Integer Linear Programs (ILP).  Computational experiments
showed, however, that the performance of the ILP scales not very favorably,
thus limiting exact ILP solutions in practice to moderate-sized data. Fast
and accurate heuristics for cograph-editing are therefore of immediate
practical interest in the field of phylogenomics.

The cotree of a cograph coincides with the modular decomposition tree
\cite{gallai-67}, which is defined for all graphs. We investigate here how
edge editing on an approximate cograph is related to editing the
corresponding modular decomposition trees.

\section{Basic Definitions}

We consider simple finite undirected graphs $G=(V,E)$ without loops. The
notation $G+e$, $G-e$ and $G\DELTA e$ is is used to denote the graph
$(V,E\cup\{e\})$, $(V,E\setminus\{e\})$ and $(V,E\DELTA\{e\})$,
respectively.  A graph $H=(W,E')$ is a \emph{subgraph} of a graph
$G=(V,E)$, in symbols $H\subseteq G$, if $W\subseteq V$ and $E'\subseteq
E$.  If $H \subseteq G$ and $xy \in E'$ if and only if $xy\in E$ for all
$x,y\in W$, then $H$ is called an \emph{induced} subgraph.  We often denote
such an induced subgraph $H=(W,E')$ by $G[W]$.  A \emph{connected
  component} of $G$ is a connected induced subgraph that is maximal w.r.t.\
inclusion. The complement $\overline G$ of a graph $G=(V,E)$ has vertex set
$V$ and edge set $E(\overline G)=\{xy\mid x,y\in V, x\neq y, xy\notin E\}$.
The \emph{complete graph} $K_{|V|}=(V,E)$ has edge set $E={V\choose{2}}$.
We write $G\simeq H$ for two isomorphic graphs $G$ and $H$.
   
Let $G=(V,E)$ be a graph. The \emph{(open) neighborhood} $N(v)$ is defined
as $N(v)=\{x\mid vx\in E\}$. The \emph{(closed) \emph{neighborhood}} $N[v]$
is then $N[v] = N(v)\cup\{v\}$. If there is a risk of confusion we will
write $N_G(v)$, resp., $N_G[v]$ 
to indicate that the respective neighborhoods are taken w.r.t.\ $G$. The
\emph{degree} $\deg(v)$ of a vertex is defined as $\deg(v) = |N(v)|$.

A \emph{tree} is a connected graph that does not contain cycles. A
\emph{path} is a tree where every vertex has degree $1$ or $2$. A
\emph{rooted} tree $T=(V,E)$ is a tree with one distinguished vertex
$\rho\in V$. The first inner vertex $\lca(x,y)$ that lies on both unique
paths from two vertices $x$, resp., $y$ to the root, is called \emph{lowest
  common ancestor} of $x$ and $y$.  It is well-known that there is a
one-to-one correspondence between (isomorphism classes of) rooted trees on
$V$ and so-called hierarchies on $V$.  For a finite set $V$, a
\emph{hierarchy on $V$} is a subset $\mathcal C$ of the power set $\mathcal
P(V)$ such that $(i)$ $V\in \mathcal{C}$, $(ii)$ $\{x\}\in \mathcal{C}$ for
all $x\in V$ and $(iii)$ $p\cap q\in \{p, q, \emptyset\}$ for all $p, q\in
\mathcal{C}$.
\begin{theorem}[\cite{sem-ste-03a}]
  Let $\mc C$ be a collection of non-empty subsets of $V$.  Then, there is
  a rooted tree $T=(W,E)$ on $V$ with $\mc C = \{L(v)\mid v\in W\}$ if
  and only if $\mc C$ is a hierarchy on $V$.
  \label{A:thm:hierarchy}
\end{theorem}

\section{Cographs, $\boldsymbol{P_4}$-sparse Graphs and the Modular
  Decomposition}

\subsection{Introduction to Cographs}

\emph{Cographs} are defined as the class of graphs formed from a single
vertex under the closure of the operations of union and complementation,
namely: (i) a single-vertex graph $K_1$ is a cograph; (ii) the disjoint
union $G=(V_1\cup V_2,E_1\cup E_2)$ of cographs $G_1=(V_1,E_1)$ and
$G_2=(V_2,E_2)$ is a cograph; (iii) the complement $\overline{G}$ of a
cograph $G$ is a cograph.  The name cograph originates from
\emph{complement reducible graphs}, as by definition, cographs can be
``reduced'' by stepwise complementation of connected components to totally
disconnected graphs \cite{Seinsche1974}.

It is well-known that for each induced subgraph $H$ of a cograph $G$ either
$H$ is disconnected or its complement $\overline H$ is disconnected
\cite{BLS:99}.  This, in particular, allows representing the structure of a
cograph $G=(V,E)$ in an unambiguous way as a rooted tree $T=(W,F)$, called
\emph{cotree}: If the considered cograph is the single vertex graph $K_1$,
then output the tree $(\{u\}, \emptyset)$. Else if the given cograph $G$ is
connected, create an inner vertex $u$ in the cotree with label ``series'',
build the complement $\overline G$ and add the connected components of
$\overline G$ as children of $u$. If $G$ is not connected, then create an
inner vertex $u$ in the cotree with label ``parallel'' and add the
connected components of $G$ as children of $u$.  Proceed recursively on the
respective connected components that consists of more than one
vertex. Eventually, this cotree will have leaf-set $V\subseteq W$ and inner
vertices $u\in W\setminus V$ are labeled with either ``parallel'' or
``series'' s.t.\ $xy\in E$ if and only if $u=\lca_T(x,y)$ is labeled
``series''. Since a cograph and its cotree are uniquely determined by each
other, one can use the cotree representation to test in linear-time whether two
cographs are isomorphic \cite{Corneil:85}.

The complement of a path on four vertices $P_4$ is again a $P_4$ and hence,
such graphs are not cographs.  Intriguingly, cographs have indeed a quite
simple characterization as \emph{$P_4$-free} graphs, that is, no four
vertices induce a $P_4$.  A number of further equivalent characterizations
are given in \cite{BLS:99} and Theorem \ref{thm:cograph-characterization}.
Determining whether a graph is a cograph can be done in linear time
\cite{Corneil:85,BCHP:08}.

\subsection{Modules and the Modular Decomposition}

The concept of \emph{modular decompositions (MD)} is defined for arbitrary
graphs $G$. It present the structure of $G$ in the form of a tree that
generalizes the idea of cotrees. However, in general much more information
needs be stored at the inner vertices of this tree if the original graph is
to be recovered. 

The MD is based on modules, which are also known as autonomous sets 
\cite{MR-84, Moh:85}, closed sets \cite{gallai-67}, clan \cite{EGMS:94}, 
stable sets, clumps \cite{Blass:78} or externally related sets \cite{HM-79}.
A \emph{module} of a given graph $G = (V,E)$ is
a subset $M\subseteq V$ with the property that for all vertices in $x,y\in
M$ holds that $N(y)\setminus M = N(x)\setminus M$. The vertices within a
given module $M$ are therefore not distinguishable by the part of their
neighborhoods that lie ``outside'' $M$. Modules can thus be seen as
generalization of the notion of connected components.  We denote with
$\M(G)$ the set of all modules of $G=(V,E)$. Clearly, the vertex set $V$
and the singletons $\{v\}$, $v\in V$ are modules, called \emph{trivial}
modules. A graph $G$ is called \emph{prime} if it only contains trivial
modules.  For a module $M$ of $G$ and a vertex $v\in M$, we define the
\out{M}-neighborhood of $v$ in $G$ as the set $N^M_G(v):=N_{G}(v)\setminus
M$. Since for each $v,w\in M$ the \out{M}-neighborhoods are identical,
$N^M_G(v)=N^M_G(w)$, we can equivalently define the \out{M}-neighborhood of
the module $M$ as $N^M_G:=N^M_G(v),\ v\in M$.

For a graph $G=(V,E)$ let $M$ and $M'$ be disjoint subsets of $V$. We say
that $M$ and $M'$ are adjacent (in $G$) if each vertex of $M$ is adjacent
to all vertices of $M'$; the sets are non-adjacent if none of the vertices
of $M$ is adjacent to a vertices of $M'$. Two disjoint modules are either
adjacent or non-adjacent \cite{Moh:85}. One can therefore define the
\emph{quotient graph} $G/\M'$ for an arbitrary subset $\M'\subseteq \M(G)$
of pairwise disjoint modules: $G/\M'$ has $\M'$ as its vertex set and
$(M_i,M_j)\in E(G/\M')$ if and only if $M_i$ and $M_j$ are adjacent in $G$.

A module $M$ is called \emph{strong} if for any module $M'\ne M$ either $M
\cap M' = \emptyset$, or $M \subseteq M'$, or $M \subseteq M'$, i.e., a
strong module does not overlap any other module. The set of all strong
modules $\MD(G)\subseteq \M(G)$ thus forms a hierarchy, the so-called
\emph{modular decomposition} of $G$. While arbitrary modules of a graph
form a potentially exponential-sized family, however, the sub-family of
strong modules has size $O(|V(G)|)$ \cite{habib2004simple}.

Let $\P=\{M_1, \dots, M_k\}$ be a partition of the vertex set of a graph
$G=(V,E)$. If every $M_i\in \P$ is a module of $G$, then $\P$ is a modular
partition of $G$. A non-trivial modular partition $\P=\{M_1, \dots, M_k\}$
that contains only maximal (w.r.t\ inclusion) strong modules is a maximal
modular partition.  We denote the (unique) maximal modular partition of $G$
by $\Pmax(G)$. We will refer to the elements of $\Pmax(G[M])$ as the
\emph{the children of $M$}. This terminology is motivated by the following
  considerations:

The hierarchical structure of $\MD(G)$ gives rise to a canonical tree
representation of $G$, which is usually called the \emph{modular
  decomposition tree} $\MDT(G)$ \cite{MR-84,HP:10}. The root of this tree is
the trivial module $V$ and its $|V|$ leaves are the trivial modules
$\{v\}$, $v \in V$. The set of leaves $L_v$ associated with the subtree
rooted at an inner vertex $v$ induces a strong module of $G$. Moreover,
inner vertices $v$ are labeled ``parallel'' if the induced subgraph
$G[L_v]$ is disconnected, ``series'' if the complement $\overline{G[L_v]}$
is disconnected and ``prime'' otherwise, i.e., if $G[L_v]$ and
$\overline{G[L_v]}$ are both connected. The module $L_v$ of the induced
subgraph $G[L_v]$ associated to a vertex $v$ labeled ``prime'' is called
prime module. Note, the latter does not imply that $G[L_v]$ is prime,
however, in all cases $G[L_v]/\Pmax(G[L_v])$ is prime \cite{HP:10}.
Similar to cotrees it holds that $xy\in E$ if $u=\lca_{\MDT(G)}(xy)$ is
labeled ``series'', and $xy\notin E$ if $u=\lca_{\MDT(G)}(xy)$ is labeled
``parallel''. However, to trace back the full structure of a given graph
$G$ from $\MDT(G)$ one has to store additionally the information of the
subgraph $G[L_v]/\Pmax(G[L_v])$ in the vertices $v$ labeled ``prime''.
Although, $\MD(G)\subseteq \M(G)$ does not represent all modules, we state
the following remarkable fact \cite{Moh:85,DGC:1997}: Any subset
$M\subseteq V$ is a module if and only if $M\in \MD(G)$ or $M$ is the union
of children of non-prime modules.  Thus, $\MDT(G)$ represents at least
implicitly all modules of $G$.

A simple polynomial time recursive algorithm to compute $\MDT(G)$ is as follows
\cite{HP:10}: 
(1) compute the maximal modular partition \Pmax(G); (2) 
label the root node according to the parallel, series or prime type of $G$; 
(3) for each strong module $M$ of $\Pmax$, compute $\MDT(G[M])$ 
and attach it to the root node.
The first polynomial algorithm to compute the modular decomposition 
is due to Cowan \emph{et al.\ }\cite{CJS:72}, and it runs in $O(|V|^4)$.
Improvements are due to Habib and Maurer \cite{HM-79}, who proposed
a cubic time algorithm, and to M\"{u}ller and
Spinrad \cite{MS:89}, who designed a quadratic time algorithm. The
first two linear time algorithms appeared independently in 1994 \cite{CH:94, CS94}. 
Since then a series of simplified algorithms has been published, some
running in linear time \cite{DGC:01,CS:99,TCHP:08}, and others in almost linear time 
\cite{DGC:01,CS:00,HPV:00, habib2004simple}. 

We give here two simple lemmata for further reference.

\begin{lemma}
  Let $M$ be a module of a graph $G=(V,E)$ and $M'\subseteq M$. Then $M'$
  is a module of $G[M]$ if and only if $M'$ is a module of $G$.\\
  Furthermore, suppose $M\in \MD(G)$ is strong module of $G$. Then $M'$ is
  a strong module of $G[M]$ if and only if $M'$ is a strong module of $G$.
\label{lem:module-subg}
\end{lemma}
\begin{proof}
  Let $M\in \M(G)$. If $M'$ is a module of $G[M]$, then all $x,y \in M'$
  have the same \out{M'}-neighbors in $G[M]$.  Since $M$ is a module of $G$
  and $M'\subseteq M$, for all $x,y \in M'$ the \out{M'}-neighborhood and
  \out{M}-neighborhood in $G[V\setminus M]$ are identical.  Thus, all $x,y
  \in M'$ have the same \out{M'}-neighborhood in $G$.

  If $M'\subseteq M$ is a module in $G$ then, in particular, the
  \out{M'}-neighborhood in $G[M]$ must be identical for all $x,y\in M'$,
  and thus $M'$ is a module in $G[M]$.

  Let $M\in \MD(G)$ and assume that $M'$ is a strong module of $G[M]$.
  Since $M$ is a strong module in $G$ it does not overlap any other modules
  in $G$. Assume for contradiction that $M'$ is not a strong module of $G$.
  Hence $M'$ must overlap some module $M''$ in $G$. This module $M''$
  cannot be entirely contained in $M$ as otherwise, $M''$ and $M'$ overlap
  in $G[M]$ implying that $M'$ is not a strong module of $G[M]$, a
  contradiction.  But then $M$ and $M''$ must overlap, contradicting that
  $M\in \MD(G)$.

  If $M'$ is a strong module of $G$ then it does not overlap \emph{any}
  module of $G$. As every module of $G$ is also a module of $G[M]$ (and
  vice versa) it follows that $M'$ does not overlap \emph{any} module of
  $G[M]$ and thus, $M'$ must be a strong module of $G[M]$.
    \end{proof}

\begin{lemma}
  Let $G$ be an arbitrary graph and $G'$ be a cograph
  on the same vertex set $V$   so that $\M(G)\subset
  \M(G')$, i.e., every module of $G$ is a module of $G'$. Moreover, let
  $\Pmax:=\Pmax(G)$ be the maximal modular partition of $G$.  Then $\Pmax$
  is a modular partition of $G'$ and $G'/\Pmax$ is a cograph.
  \label{lem:quotient-cograph}
\end{lemma}
\begin{proof}
  Since $\M(G)\subset \M(G')$ we can immediately conclude that $\Pmax$ is a
  (not necessarily maximal) modular partition of $G'$ and therefore the
  quotient $G'/\Pmax$ is well-defined.  Assume, for contradiction, that
  $G'/\Pmax$ is not a cograph.  Then $G'/\Pmax$ must contain one induced
  $P_4$, say $M_1-M_2-M_3-M_4$.  As $M_1,\ldots, M_4$ are modules of $G'$
  and since two disjoint modules are either adjacent or non-adjacent it
  follows that $G'$ must contain an induced $P_4$ of the form
  $x_1-x_2-x_3-x_4$ with $x_i\in M_i$, $1\leq i\leq 4$, a contradiction.
    
\end{proof}

\subsection{The Twin-Relation}	

A special kind of module that will play a central role in this contribution
are \emph{twins}. Two vertices $x,y\in V$ are called twins if $\{x,y\}$ is
a module of $G$. Twins $x,y\in V$ are called \emph{true twins} if $xy\in E$
and \emph{false twins} otherwise. Twins $x$ and $y$ therefore satisfy
$N(x)\setminus\{y\} = N(y)\setminus\{x\}$. In particular, for true twins
$\{x,y\}$ we can infer that $N[x]=N[y]$ and for false twins $\{x,y\}$ we
only have $N(x) = N(y)$. 

\begin{definition}
  Let $G = (V,E)$ be an arbitrary graph.  The \emph{twin relation} $\mc T$
  is the binary relation on $V$ that contains all pairs of twins: $(x,y)\in
  \mc T$ if and only if $x,y$ are twins in $G$.  The pair $(x,x)\in \mc T$
  is called \emph{trivial twin}.  
\end{definition}
Unless explicitly stated, we will use the phrase ``a pair of twins'' or 
``twins'', for short, to refer only to non-trivial twins.

\begin{proposition}		
  Let $G=(V,E)$ be a given graph and  $\mc T$ the twin relation on $V$.
  Then the following statements hold:
  \begin{enumerate}
  \item The relation $\mc T$ is an equivalence relation. 
  \item For every equivalence class $M\eqcl \mc T$ 
    the distinct elements of $M$ are either all true or false twins. 
  \item Every equivalence class $M\eqcl T$ is a module of $G$
    and there is no other non-trivial strong module contained in $M$. 
  \end{enumerate}
  \label{prop:twin}
\end{proposition}
\begin{proof}
  \emph{We first prove Statement 1. and 2.:} Clearly, $\mc T$ is reflexive
  and symmetric. It remains to show that $\mc T$ is transitive.  Assume
  that $(x,y),(x,z) \in \mc T$. We show first that $x,y$ and $x,z$ can only
  be either true or false twins. Assume for contradiction that $x,y$ are
  false twins and $x,z$ are true twins. Hence, $(x,z)\in E(G)$ and since
  $N(x)=N(y)$ we have $(y,z)\in E(G)$. However, since $y\in N[z] = N[x]$ we
  have $y\in N(x)$ and hence $(x,y)\in E(G)$, a contradiction.
		
  Let $(x,y),(x,z) \in \mc T$ be both false twins. Thus, $N(z)= N(x) =
  N(y)$, which implies that $(y,z) \in \mc T$. Moreover, since $N(z)=
  N(y)$, the vertices $y$ and $z$ cannot be adjacent and thus, $x$ and $z$
  are false twins. Now assume that $(x,y),(x,z) \in \mc T$ are both true
  twins. Hence, $N[z] = N[x] = N[y]$ and, thus, $(y,z) \in
  E(G)$. Therefore, $(y,z)\in \mc T$ are true twins.

  Thus, $T$ is transitive and hence, an equivalence relation, where each
  equivalence class $M\eqcl \mc T$ comprises either only false or only true
  twins.

  \emph{Now, we prove Statement 3.:} Statement 1.\ and 2.\ imply that
  $M\eqcl \mc T$ contains either only true or false twins.  If $M$ contains
  only false twins, then $N(x)\setminus M = N(x) = N(y) = N(y)\setminus M$
  for all $x,y\in M$ and thus, $M$ is a module. If $M$ contains only true
  twins $x,y$, then $N(x)\setminus \{y\} = N(y)\setminus \{x\}$. If there
  is an additional vertex $z\in M$ we must have $N(x)\setminus \{y,z\}
  =N(y)\setminus \{x,z\}$. Induction on the number of elements of $M$ shows
  that $N(x)\setminus M = N(y)\setminus M$ holds for all $x,y\in M\eqcl\mc
  T$. Therefore, $M$ is a module.

  Finally, assume there is a non-trivial strong module $M'$ contained in
  $M\eqcl \mc T$.  Let $x\in M'\subsetneq M$ and $z\in M\setminus
  M'$. Since $x,z\in M\eqcl\mc T$, they are twins. Thus, $\{x,z\}$ is a
  module. Hence, $M'$ cannot be strong module since $\{x,z\}\cap M'=\{x\}$
  and thus, $\{x,z\}$ and $M'$ overlap.  
    
\end{proof}

Note that equivalence classes of $\mc T$ are not necessarily strong
modules, as the following example shows. Consider the graph
$G=(\{0,1,2,3\},\{(2,3)\})$.  The twin relation $\mc T$ on $G$ has
equivalence classes $\{0,1\}$ and $\{2,3\}$.  However, $M =\{1,2,3\}$ is
also a module, as $N(i)\setminus M=\emptyset$, $1\leq i\leq 3$.  In this
case, $M$ overlaps $\{0,1\}$.  The modular decomposition $\MD(G)$ is
$\{\{0\}, \{1\},\{2\},\{3\}, \{2,3\}, \{0,1,2,3\}\}$, while $\M(G)\setminus
\MD(G) = \{\{0,2,3\}, \{1,2,3\}\}$.

\subsection{$\boldsymbol{P_4}$-sparse Graphs and Spiders}

Although the cograph-editing problem is NP-complete, it can be solved in
polynomial time for so-called \emph{$P_4$-sparse} graphs \cite{Liu:11,
  Liu:12}, in which every set of five vertices induces at most one $P_4$
\cite{Hoang:85}. The efficient recognition of $P_4$-sparse graphs is
intimately connected to so-called spider graphs, a very peculiar class of
prime graphs.

\begin{lemma}\cite{JO:89,Jamison:92}.
A graph $G$ is $P_4$ sparse if and only if exactly one of the following three
alternatives is true for every induced subgraph $H$ of $G$: (i) $H$ is not
connected, (ii) $\overline{H}$ is not connected, or (iii) $H$ is a spider.
\label{lem:p4sparse-spider}
\end{lemma}

Spiders come in two sub-types, called \emph{thin} and \emph{thick}
\cite{Jamison:92,Nastos:12}. A graph $G$ is a \emph{thin spider} if its
vertex set can be partitioned into three sets $K$, $S$, and $R$ so that (i)
$K$ is a clique; (ii) $S$ is a stable set; (iii) $|K|=|S|\ge 2$; (iv) every
vertex in $R$ is adjacent to all vertices of $K$ and none of the vertices
of $S$; and (v) each vertex in $K$ is connected to exactly one vertex in
$S$ by an edge and \emph{vice versa}. A graph $G$ is a \emph{thick spider}
if its complement $\overline{G}$ is a thin spider.  The sets $K$, $S$, and
$R$ are usually referred to as the \emph{body}, the set of \emph{legs}, and
\emph{head}, resp., of a thin spider.  The path $P_4$ is the only graph
that is both a thin and thick spider. Interestingly, spider graphs are
fully characterized by its degree sequences \cite{BCF+15}.

Lemma \ref{lem:p4sparse-spider} in particular implies that any strong
module $M\in \MD(G)$ of a $P_4$-sparse graph is either (i) parallel (ii)
series or (iii) prime, in which case it corresponds to a spider $G[M]$.  In
general, we might therefore additionally distinguish prime modules $M$ as
those where $G[M]$ is a spider (called \emph{spider modules}) and those
where $G[M]$ is not a spider (which we still call simply \emph{prime
  modules}).

\subsection{Useful Properties of Modular Partitions}

First, we briefly summarize the relationship between cographs $G$ and the
modular decomposition $\MD(G)$.
\begin{theorem}[\cite{Corneil:81,BLS:99}]
Let $G=(V,E)$ be an arbitrary graph. Then the following statements are
equivalent. 
\begin{enumerate}
\item $G$ is a cograph.  
\item $G$ does not contain induced paths on four vertices $P_4$.
\item $\MDT(G)$ is the cotree of $G$ and hence, has no inner vertices
  labeled with ``prime''.
\item Any non-trivial induced subgraph of $G$ has at least one pair of
  twins $\{x,y\}$.
\label{thm:cograph-characterization}
\end{enumerate}
\label{thm:cograph-characterize}
\end{theorem}

For later explicit reference, we summarize in the next theorem several
results that we already implicitly referred to in the discussion above.

\begin{theorem}[\cite{GPP:10,HP:10,Moh:85}]
\label{thm:all}
The following statements are true for an arbitrary graph $G=(V,E)$:
\begin{itemize}
\item[(T1)]
  The maximal modular partition $\Pmax(G)$ and the modular decomposition
  $\MD(G)$ of $G$ are unique.
\item[(T2)] 
  Let $\P$ be a modular partition of $G$.  Then $\widetilde{\P} \subset \P$
  is a (non-trivial strong) module of $G/{\P}$ if and only if $\cup_{M\in
    \widetilde{\P}} M$ is a (non-trivial strong) module of $G$.
\item[(T3)] 
  Let $M$ be a module of $G$ and $\{a,b,c,d\}$ be four vertices inducing a
  $P_4$ in $G$, then $|M \cap \{a,b,c,d\}|\leq 1$ or $\{a,b,c,d\}\subseteq
  M$.
\item[(T4)] 
  For any connected graph $G$ with $\overline G$ being connected, the
  quotient $G/{\Pmax(G)}$ is a prime graph.
\item[(T5)] Let $\Pmax$ be the maximal modular partition of $G[M]$, where
  $M$ denotes a prime module of $G$ and $\P' \subsetneq \Pmax$ be a proper
  subset of $\Pmax$ with $|\P'|>1$.  Then, $\bigcup_{M'\in \P'}{M'} \notin
  \M(G)$.
\end{itemize}
\end{theorem}

Statements (T1) and (T4) are clear.  Statement (T2) characterizes the
(non-trivial strong) module of $G$ in terms of (non-trivial strong) modules
$\widetilde{\P}$ of $G/{\P}$. Statement (T3) clarifies that each induced
$P_4$ is either entirely contained in a module or intersects a module in at
most one vertex.  Statement (T5) explains that none of the unions of
elements of a maximal modular partition of $G[M]$ are modules of
$G$. Hence, only the prime module $M$ itself and, by Lemma
\ref{lem:module-subg}, the elements $M'\in \Pmax$ are modules of $G$.
 
\section{Cograph Editing}

\subsection{Optimal Modul-Preserving Edit Sets}

Given an arbitrary graph we are interested in the following
optimization problem. 
\begin{problem}[Cograph Editing]
  Given a graph $G=(V,E)$. Find a set $F\subseteq {{V}\choose{2}}$ of
  minimum cardinality s.t.\ $G^*=(V,E\DELTA F)$ is a cograph.
\end{problem} 
We will simply call an edit sets of minimum cardinality an \emph{optimal
  edit set}.

The (decision version of the) cograph-editing problem is NP-complete
\cite{Liu:11,Liu:12}. Nevertheless, the cograph-editing problem is
fixed-parameter tractable (FPT) \cite{Protti:09}. Hence, for the
parametrized version of this problem, i.e., for a given graph $G = (V,E)$
and a parameter $k \geq 0$ find a set $F$ of at most $k$ edges and
non-edges so that $G^* = (V, E\DELTA F )$ is a cograph, there is an
algorithm with running time $O(6^k)$ \cite{Cai:96}. This FPT approach was
improved in \cite{Liu:11,Liu:12} to an $O(4.612^k + |V|^{4.5})$ time
algorithm. These results are of little use for practical applications,
because the constant $k$ can become quite large. However, they provide deep
insights into the structure of the class of $P_4$-sparse graphs that
slightly generalizes cographs.

In particular, the cograph-editing problem can be solved in polynomial time
whenever the input graph is $P_4$-sparse \cite{Liu:11,Liu:12}. The key
observation is that every strong \emph{prime} module $M$ of a $P_4$-sparse
graph $G$ is a spider module. The authors then proceed to show that it
suffices to edit a fixed number of (non)legs, i.e., only (non)edges $xy$
with $x\in K$ and $y\in S$ in $G[M]$, for all such spider modules to
eventually obtain an optimally edited cograph. The resulting algorithm to
optimally edit a $P_4$-sparse graph to a cograph, \texttt{EDP4}, runs in
$O(|V|+|E|)$-time.

In the following will frequently make use of a result by Guillemot \emph{et
  al.\ } \cite{GPP:10} that is based on the following
\begin{lemma}[\cite{GPP:10}]
  Let $G = (V,E)$ be an arbitrary graph and let $M$ be a non-trivial module
  of $G$. If $F_M$ is an optimal edge-edition set of the induced subgraph
  $G[M]$ and $\Fo$ is an optimal edge-edition set of $G$, then 
  (i) $F = (\Fo\setminus\Fo[M])\cup F_M$ is an optimal edge-edition set of
  $G$ and (ii) $\Fo[M]$ contains all (non-)edges $xy\in \Fo$ with $x,y\in M$.
  \label{lem:opt-modul-edit}
\end{lemma}

\begin{proposition}[\cite{GPP:10}]
  Every graph $G(V,E)$ has an optimal edit set $\Fo$ such that every module
  $M$ of $G$ is module of the cograph $\Go = (V,E \DELTA \Fo)$.
\label{prop:FM}
\end{proposition}
An edit set as described in Prop.~\ref{prop:FM} is called
\emph{module-preserving}. Their importance lies in the fact that
module-preserving edit sets update either all or none of the edges between
any two disjoint modules. 

In the following Remark we collect a few simple consequences of our
considerations so far.
\begin{remark}
  By Theorem \ref{thm:all} (T3) and definition of cographs, all induced
  $P_4$'s of a graph are entirely contained in the prime modules.  By Lemma
  \ref{lem:module-subg}, the maximal modular partition $\Pmax$ of $G[M]$ is
  a subset of the strong modules $\MD(G)\subseteq \M(G)$, for all strong
  modules $M\in \MD(G)$.  Taken together with Lemma
  \ref{lem:opt-modul-edit}, Proposition \ref{prop:FM} and Theorem
  \ref{thm:all} $(T1)$, this implies that it suffices to solve the
  cograph-editing problem for $G$ independently on each of $G$'s strong
  prime modules \cite{GPP:10}.

An optimal module-preserving edit-set $\Fo$ on $G$ therefore induces optimal
edit-sets $F_M$ on $G[M]$ for any $M$, and thus also optimal edit-set
$\Fo(M,\Pmax)$ on $G[M]/{\Pmax}$, where $\Pmax = \{M_1,\dots,M_k\}$ is again
the maximal modular partition of $G[M]$ and $M$ is a module of $\M(G)$. The
edit set $F(M,\Pmax)$ has the following explicit representation:
\[
F(M,\Pmax):=\{\{M_i,M_j\}\mid M_i,M_j\in\Pmax \exists x\in M_i, y\in M_j 
\textrm{\ with\ } \{x,y\}\in \Fo[M]\}.
\]
\label{rem:all}
\end{remark}

\subsection{Optimal Module Merge Deletes All $P_4$'s}

Since cographs are characterized by the absence of induced $P_4$'s, we can
interpret every cograph-editing method as the removal of all $P_4$'s in the
input graph with a minimum number of edits. A natural strategy is therefore
to detect $P_4$'s and then to decide which ones must be edited. Optimal
edit sets are not necessarily unique. A further difficulty is that editing
an edge of a $P_4$ can produce new $P_4$'s in the updated graph. Hence we
cannot expect \emph{a priori} that local properties of $G$ alone will allow
us to identify optimal edits.

By Remark \ref{rem:all}, on the other hand, it is sufficient to edit within
the prime modules. We therefore focus on the maximal modular partition
$\Pmax=\Pmax(G[M])$ of $G[M]$, where $M\in \MD(G)$ is a strong prime module
of $G$. Since $G[M]/\Pmax$ is prime, it does not contain any twins.  Now
suppose we have edited $G$ to a cograph $\Go$ using an optimally
module-preserving edit set. Then $\Go[M]$ is a cograph and by Lemma
\ref{lem:quotient-cograph} the quotient $\Go[M]/\Pmax$ is also a cograph.
Therefore, $\Go[M]/\Pmax$ contains at least one pair of twins
$\{M_i,M_j\}$, where $M_i$ and $M_j$ are, by construction, children of the
prime module $M\in \MD(G)$.

This consideration suggests that it might suffice to edit the \out{M_i}-
and \out{M_j}-neighborhoods in $G$ in such a way that $M_i$ and $M_j$
become twins in an optimally edited cograph $\Go$. In the following we will
show that this is indeed the case.

We first show that twins are ``safe'', i.e., that we never have to edit
edges within a subgraph $G[M]$ induced by an equivalence classes $M$ of the
twin relation $\T$.

\begin{lemma}
  Let $G=(V,E)$ be a non-cograph, $F$ be an arbitrary cograph edit set s.t.\ 
  $G'=(V,E\DELTA F)$ is the resulting cograph and suppose that 
  $G'\DELTA e$ is a non-cograph for all $e\in F$.
  Then $\{x,y\}\notin F$ for twins $x,y$ in $G'$.
  \label{lem:T-notEdit}
\end{lemma}
\begin{proof}
  Since $G'$ is a cograph it contains at least one pair of twins
  $x,y$. First assume that $x$ and $y$ are false twins and thus, $xy\not\in
  E(G')$. Assume, for contradiction, that $xy\in F$. By assumption, $G'+xy$
  is not a cograph and thus there is an induced $P_4$ containing the edge
  $xy$ in $G'+xy$. All such $P_4$'s that contain the edge $xy$ are (up to
  symmetries and isomorphism) of the form (i) $a-x-y-b$ or (ii)
  $x-y-b-a$. Since $x$ and $y$ are false twins in $G'$, we have
  $N_{G'}(x)=N_{G'}(y)$ and hence, there must be an edge $xb\in E(G')$ and
  therefore, $xb\in E(G'+xy)$. But this implies that $G'+xy$ is still
  cograph, the desired contradiction.
	
  Now suppose that $x$ and $y$ are true twins, i.e., $xy\in E(G')$. 
  Assume, for contradiction, that $xy\in F$ and thus $G'-xy$ is not a 
  cograph. Hence there must be an induced $P_4$ containing $x$ and $y$ 
  in $G'-xy$. All such $P_4$'s containing $x$ and $y$ are 
  (up to symmetries and isomorphism) of the form (i) $x-a-y-b$ or
  (ii) $x-a-b-y$. Since $x$ and $y$ are true twins in $G'$, we have
  $N_{G'}(x)\setminus \{y\}=N_{G'}(y)\setminus \{x\}$. This implies that
  there is the the edge $xb \in E(G' - xy)$ in both case (i) and (ii). 
  Therefore $G'-xy$ is a cograph, a contradiction. 
    
\end{proof}

\begin{theorem}
  Let $G=(V,E)$ be an arbitrary graph, $F_{\opt}$ be an optimal cograph 
  edit set for $G$, $\Go=(V,E\DELTA F_{\opt})$ the resulting cograph and 
  $\T$ be the twin relation on $\Go$. 
  Then for each equivalence class $M\eqcl \T$ 
  it holds that:
  \begin{enumerate}
  \item[(i)] $M$ is a module of $\Go$ that does not contain any other 
    non-trivial strong module of $\M(\Go)$. 
  \item[(ii)] $\Go[M]$ induces either an independent set $\overline{K_{|M|}}$ 
    or a complete graph $K_{|M|}$. 
  \item[(iii)] For all $x,y \in M$ it holds that $\{x,y\}\notin F_{\opt}$	
    and thus, $G[M] \simeq \Go[M]$.
  \end{enumerate}
  \label{thm:twins}
\end{theorem}
\begin{proof}
  Statements (i) and (ii) are an immediate consequences of Proposition
  \ref{prop:twin}.  
  If $e\in \Fo$ then $\Go\DELTA e$ is a non-cograph; otherwise
  $\Fo\setminus \{e\}$ would be an edit set with smaller cardinality,
  contradicting the optimality of $\Fo$.  Thus we can apply Lemma
  \ref{lem:T-notEdit} to infer statement (iii).  
  
\end{proof}

\begin{corollary}
  Let $G = (V,E)$ be a prime graph, 
  $F_{\opt}$ be an optimal cograph edit set,
  $\Go=(V,E\DELTA F_{\opt})$  the resulting cograph and 
  $\T$ be the twin relation on $\Go$. 
  Then $V\not\eqcl \T$.
  \label{cor:twinClassNotPrime}
\end{corollary}
\begin{proof}
  As $G = (V,E)$ is a prime graph we know that $\Fo\neq \emptyset$.  If
  $V\eqcl T$, then Theorem \ref{thm:twins} implies that $\{x,y\}\not \in
  \Fo$ for all $x,y\in V$ and hence, $\Fo= \emptyset$, a contradiction.
\end{proof}

The following definitions are important for the concepts for
the ``module merge process'' that we will extensively use
in our approach. 

\begin{definition}[Module Merge]
  Let $G$ and $H$ be arbitrary graphs on the same vertex set $V$
  with their corresponding sets of all modules $\M(G)$ and $\M(H)$, resp.
  We say that a subset $M'=\{M_1,\dots,M_k\}\subseteq \M(G)$ of modules is
  \emph{merged (w.r.t.\ $H$)} -- or, equivalently, the modules in $M'$ are
  \emph{merged (w.r.t.\ $H$)} -- if (a) each of the modules in $M'$ is a
  module of $H$, and (b) the union of all modules in $M'$ is a module of
  $H$ but not of $G$.  More formally, the modules in $M'$ are \emph{merged
    (w.r.t.\ $H$)}, if
  \begin{itemize}\setlength{\itemsep}{0pt}
    \item[{(i)}] $M_1,\dots,M_k \in \M(H)$,
    \item[{(ii)}] $M = \cup_{i=1}^k M_i \in \M(H)$, and 
    \item[{(iii)}] $M\notin \M(G)$.
  \end{itemize}
  If $\{M_1,\dots,M_k\}\subseteq \M(G)$ is merged to a new module $M\in \M(H)$
  we will write this as $M_1\merge\ldots\merge M_k = \merge_{i=1}^k M_i \to M$.
\label{def:module-merge}
\end{definition}

When modules $M_1,\dots,M_k$ of $G$ are merged w.r.t.\ $H$ then all
vertices in $M = \cup_{h=1}^k M_h$ must have the same \out{M}-neighbors in
$H$, while at least two vertices $x\in M_i$, $y\in M_j$, $1\le i\neq j\le
k$ must have different \out{M}-neighbors in $G$.
	
\begin{definition}[Module Merge Edit]
  Let $G=(V,E)$ be an arbitrary graph and $F$ be an arbitrary edit set
  resulting in the graph $H=(V,E\DELTA F)$. Assume that
  $M_1,\dots,M_k\in \M(G)$ are modules that have been merged w.r.t.\ $H$
  resulting in the module $M = \cup_{i=1}^k M_i \in \M(H)$.  Then
  \begin{equation}
    F_{H}(\merge_{i=1}^k M_i \to M) = \{(x,v)\in F\ \mid\ x\in M,
    v\notin M\} 
  \end{equation} 
\end{definition}
The edit set $F_{H}(\merge_{i=1}^k M_i \to M)$ comprises exactly those
(non)edges of $F$ that have been edited so that all vertices in $M$ have
the same \out{M}-neighborhood in $H$. In particular, it contains only
(non)edge of $F$ that are not entirely contained in $G[M]$.

\begin{figure}[htbp]
  \begin{center}
    \includegraphics[viewport=116 386 488 660, clip, scale=1.1]{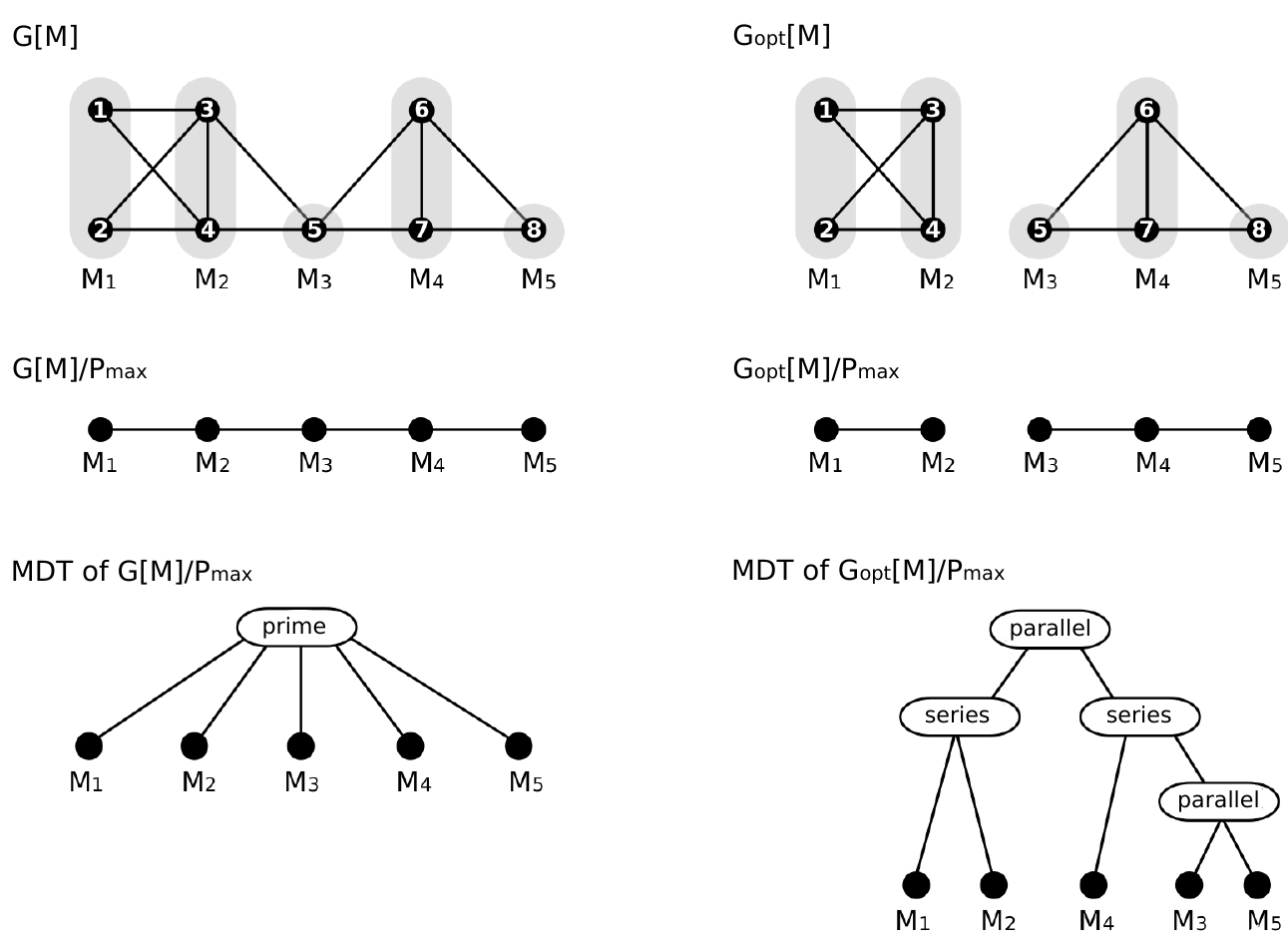}
  \caption{Assume we are given a non-cograph $G$ that contains a strong
    prime modules $M\in \MD(G)$ so that $G[M]$ is the graph in the upper
    left part of this picture.  Moreover, assume there is an optimal
    module-preserving edit set $\Fo$ transforming $G$ to a cograph $\Go$ so
    that $\{3,5\}, \{4,5\} \in \Fo$. Hence, $\Go[M]$ is the graph in the
    upper right part.  Let $\Pmax=\Pmax(G[M])=\{M_1,\ldots,M_5\}$ be the
    maximal modular partition of $G[M]$.  The modules $M_1,\ldots,M_5$
    are highlighted as gray parts in $G[M]$ and $\Go[M]$.  
    Now, $G[M]/\Pmax$ is prime
    and contains no twins, while $\Go[M]/\Pmax$ is a cograph and contains
    therefore twins.  The twin relation $\T$ on $\Go[M]/\Pmax$ has
    equivalence classes $N=\{M_1,M_2\}$, $N'=\{M_4\}$ and
    $N''=\{M_3,M_5\}$.  Hence, the modules $M_1$ and $M_2$, as well as,
    $M_3$ and $M_5$ are merged w.r.t.\ $\Go[M]/\Pmax$. Therefore, the
    modules $N_V=M_1\cup M_2 = \{1,2,3,4\}$ and $N''_V=M_3\cup M_5=\{5,8\}$ 
    have been obtained by merging
    modules of $G[M]$ w.r.t.\
    $\Go[M]$ and in particular, w.r.t.\ $\Go$. In symbols, 
    $M_1\protect\merge M_2 \to N_V$ and $M_3\protect\merge M_5 \to N''_V$.\newline
    Therefore, instead of focussing on algorithms that optimally edit
    induced $P_4$'s one can equivalently ask for optimal edit sets that
    resolve such prime modules $M$, that is, one asks for the minimum number of
    edits that adjust the neighborhoods of modules that are children of $M$
    in $\MDT(G)$ so that these modules become twins until the module $M$
    becomes a non-prime module, see Theorem \ref{thm:Edit=Merge}.  }
  \label{fig:example}
\end{center}
\end{figure}

\begin{lemma}
  Let $M$ be a strong prime module of a given graph $G=(V,E)$ 
  and $\Pmax$ be the maximal modular partition of $G[M]$.
  Moreover, let $F$ be an arbitrary
  edit set resulting in the graph $H=(V,E\DELTA F)$. 
  Assume that $\mb M=\{M_1,\ldots M_k\}\subseteq \Pmax \subseteq \M(G)$ is a set of 
  of modules that are merged w.r.t.\ $H$.

  Then for any distinct indices $a,b,c\in\{1,\dots,k\}$ it holds that
  $M_{a} \merge M_{b}\to M_{ab}$, $M_b \merge M_{a}\to M_{ba}$, and
  $M_{ab}=M_{ba}$ in $H$.  Moreover, $M_{a}\merge (M_{b} \merge M_{c})\to
  M_{a(bc)}$ and $(M_{a} \merge M_{b})\merge M_{c}\to M_{(ab)c}$ satisfy
  $M_{a(bc)}=M_{(ab)c}$, i.e., the merge operation is associative and
  commutative. The merging of any subset $\mb M$ of modules w.r.t.\ a graph
  $H$ therefore is well-defined and independent of the individual merging
  steps.

  Furthermore, 
  \begin{align*}
    F(M_{a}\merge M_{b} \merge M_{c} \to M_{abc}) = 
    & F(M_{a} \merge M_{b}\to M_{ab})\\ 
    &\cupdot \left( F(M_{ab}\merge M_{c}\to 
      M_{abc})\setminus F(M_{a} \merge M_{b}\to M_{ab})\right)
  \end{align*}
  with  $M_{a} \merge M_{b}\to M_{ab}$.
  \label{lem:pair-merge}
\end{lemma}
\begin{proof}
  For commutativity, we show first that $M_{a} \merge M_{b}\to M_{ab}$ for
  every pair of distinct $a,b\in\{1,\dots,k\}$ By definition of ``\merge''
  we have $M_{a},M_{b} \in \M(H)$, and thus, Condition (i) of Def.\
  \ref{def:module-merge} is satisfied.  Moreover, Thm.\ \ref{thm:all} (T5)
  implies that $M_{a}\cup M_{b}\notin \M(G)$ and hence, Condition (iii) of
  Def.\ \ref{def:module-merge} is satisfied.  For Condition (ii) we have to
  show that $M_a\cup M_b=M_{ab}\in \M(H)$.  Assume for contradiction
  that $M_a\cup M_b\notin \M(H)$, then there is a vertex $v\in V$ so that
  $v$ is in the \out{M_a}-neighborhood, but not in the
  \out{M_b}-neighborhood w.r.t.\ $H$, or \emph{vice versa}.  However, this
  remains true if we consider $\cup_{i=1}^k M_i$, which implies that
  $\cup_{i=1}^k M_i\notin \M(H)$, and hence $\mb{M}=\{M_1,\dots,M_k\}$ 
  is not merged w.r.t.\ $H$, a contradiction.  Finally, $M_{a} \merge
  M_{b}\to M_{ab}$ implies that $M_{ab} = M_{a} \cup M_{b} = M_{b} \cup
  M_{a}=M_{ba}$, and thus $\merge$ is commutative.

  By similar arguments one shows that Condition (i), (ii), and (iii) 
  of Def.\ \ref{def:module-merge} are satisfied for $M_a\merge M_b\merge M_c$.
  Hence, since  $M_a\merge M_b\merge M_c\to M_{abc}$ implies that
  $M_{abc} = M_a\cup M_b \cup M_c$, associativity of $\merge$ follows again
  directly from the associativity of the set union.
	
  To see that the last property for $F$ is satisfied, note that
  \begin{align*}
    &F(M_a\merge M_b \merge M_c \to M_{abc}) \\
    &=  \{(x,v)\in F\ \mid\ x\in M_{abc}, v\notin M_{abc}\} \\
    &= \{(x,v)\in F\ \mid\ x\in M_{ab}, v\notin M_{ab}\}  
    \cupdot  \{(x,v)\in F\ \mid\ x\in M_{abc}\setminus M_{ab}, 
                              v\notin M_{abc}\setminus M_{ab}\} \\
    &=   F(M_{a} \merge M_{b}\to M_{ab})
    \cupdot \left( F(M_{ab}\merge M_{c}\to M_{abc})\setminus 
                   F(M_{a}\merge M_{b}\to M_{ab})             \right)
  \end{align*}
   
\end{proof}

Let $G$ be an arbitrary graph and $\Fo$ be a minimum cardinality set of
edits that applied to $G$ result in the cograph $\Go$. We will show that
every module-preserving edit set $\Fo$ can be expressed completely by means
of module merge edits.
To this end, we will consider the strong prime
modules $M\in \MD(G)$ of the given graph $G$ (in particular certain
submodules of $M$ that do not share the same out-neighborhood) and adjust
their out-neighbors to obtain new modules as long as $M$ stays a prime
module. This procedure is repeated for all prime modules of $G$, until no
prime modules are left in $G$.

As mentioned above, if $G$ is not a cograph there must be a strong prime
module $M$ in $G$. Let $\Pmax$ be the maximal modular partition of
$G[M]$. Theorem \ref{thm:all} implies that $G[M]/\Pmax$ is prime and thus,
does not contain twins.  However, if $\Fo$ is module preserving, then Lemma
\ref{lem:quotient-cograph} implies that the graph $\Go[M]/\Pmax$ is a
cograph and thus, contains non-trivial twins by Theorem
\ref{thm:cograph-characterization}.

Hence, we aim at finding particular submodules $M_1,M_2,\dots,M_k \subset
M$ in $G$ that have to be merged so that they become twins in
$\Go[M]/\Pmax$.  As the following theorem shows, repeated application of
this procedure to all prime modules of $G$ will completely resolve those
prime modules resulting in $\Go$. An illustrative example is shown in
Figure \ref{fig:example}.

For the proof of the final Theorem \ref{thm:Edit=Merge} we first establish
the following result.
\begin{lemma}
  Let $G=(V,E)$ be graph, $F_{\opt}$ be an optimal module-preserving cograph
  edit set, and $\Go=(V,E\DELTA F_{\opt})$ be the resulting cograph. 
  Furthermore, let $M$ be an arbitrary strong prime module of $G$
  that does not contain any other strong prime module,  
  and $\Pmax:=\Pmax(G[M])=\{M_1,\dots, M_k\}$
  be the maximal modular partition of $G[M]$. 
  Moreover, let 
  $F' = \{\{x,y\}\in\Fo \mid \exists M_i,M_j\in \Fo(M,\Pmax)
  \text{\ with\ } x\in M_i, y\in M_j\}$, 
  where $\Fo(M,\Pmax)$ is the edit set to implied by $\Fo$
  on  $G[M]/{\Pmax}$ defined in Remark \ref{rem:all}.	 

  Then, every strong prime module of $H = (V,E\DELTA F')$ is a strong prime
  module of $G$. Moreover, $\Pmax(H[M']) = \Pmax(G[M'])$ holds for each
  strong prime module $M'$ of $H$.
\label{lem:strongHG}
\end{lemma}
\begin{proof}	
  First note, that there is no other module $M'\subsetneq M$
  containing induced $P_4$'s because $M$ does not contain any other strong
  prime module $M'\subsetneq M$ and because by Theorem \ref{thm:all} (T5)
  the union $\cup_{M'\in P'} M'$ is not a module of $G$ for any of the
  subsets $P'\subsetneq \Pmax$. Together with Lemma
  \ref{lem:opt-modul-edit} this implies that $F'=\{\{x,y\}\in\Fo \mid
  x,y\in M\}$. Thus, $H[M]$ is a cograph as otherwise, there would be
  induced $P_4$'s contained in $H[M]$ and there are no further edits in
  $\Fo\setminus F'$ to remove these $P_4$'s. Thus, such $P_4$'s would
  remain in $\Go$, a contradiction.
 
  Let $M'$ be an arbitrary strong prime module of $H$. Note, since $F'$
  does not affect the \out{M}-neighborhood, $M$ is still a module of
  $H$. Since $M'$ is a strong module in $H$ we have, therefore, either,
  $M'\subseteq M$, $M\subsetneq M'$ or $M\cap M'=\emptyset$ in $G'$.
  
  Assume that $M'\subseteq M$. Since $M'$ is prime in $H$ it follows that
  $H[M']$ is a non-cograph, a contradiction, since $H[M']$ is an induced
  subgraph of the cograph $H[M]$. Hence, the case $M'\subseteq M$ cannot
  occur. If $M\subsetneq M'$ or $M\cap M'=\emptyset$, then the
  \out{M'}-neighborhood of any vertex contained in $M'$ have not been
  affected by $F'$, and thus, $M'$ is also a module of $G$.
 
  It remains to show for the cases $M\subsetneq M'$ or $M\cap M'=\emptyset$
  that the module $M'$ of $G$ is also strong and prime.  If $M\cap
  M'=\emptyset$, then $F'$ does not affect any vertex of $M'$ and hence,
  $M'$ must be a strong prime module of $G$ and, in particular, $\Pmax(H[M'])
  = \Pmax(G[M'])$.

  Let $M\subsetneq M'$ and assume for contradiction that $M'$ is not strong
  in $G$. Thus, $M'$ must overlap some other module $M''$ in $G$. However,
  since $M$ is a strong module of $G$, $M$ cannot overlap $M''$ and hence,
  $M\cap M''=\emptyset$. However, as $F'$ does only affect the vertices
  within $M$ and, in particular, none of the vertices of $M''$ and since
  $M'$ and $M''$ overlap in $G$, they must also overlap in $H$, a
  contradiction. Thus, $M'$ is a strong module of $G$.
  
  Furthermore, let again $M\subsetneq M'$ and assume that $M'$ is not prime
  in $G$. Now, let $\P'_{\max}$ be the maximal modular partition of
  $G[M']$. Since $M$ and $M'$ are strong modules in $G$ and $M\subsetneq
  M'$, we have either $M\in\P'_{\max}$ or $M\subsetneq M''\in\P'_{\max}$
  for some strong module $M''\in \Pmax$. Hence, all modules of $\Pmax$ are
  entirely contained in the modules of $\P'_{\max}$. In particular, this
  implies that we have not changed the \out{M''}-neighborhood for any
  $M''\in \P'_{\max}$, and therefore, the maximal modular partition of
  $G[M']$ is also the maximal modular partition of $H[M']$, i.e.,
  $\Pmax(H[M']) = \Pmax(G[M'])$.  However, if $M'$ is not prime in $G$,
  then $G[M']/\P'_{\max}$ is either totally disconnected ($M'$ is parallel)
  or a complete graph ($M'$ labeled series), while $H[M']/\P'_{\max}$ is
  prime (hence, it does contain only trivial modules). Clearly, if
  $G[M']/\P'_{\max}$ is totally disconnected or a complete graph, any
  subset $P'\subseteq P$ forms a module and thus, $G[M']/\P'_{\max}$ does
  not contain only trivial modules. In summary, $G[M']/\P'_{\max}\not\simeq
  H[M']/\P'_{\max}$, a contradiction, since we have not changed the
  \out{M''}-neighborhood for any $M''\in \P'_{\max}$.   
\end{proof}

\begin{theorem}
  Let $G=(V,E)$ be graph, $F_{\opt}$ be an optimal module-preserving cograph
  edit set, and $\Go=(V,E\DELTA F_{\opt})$ be the resulting cograph.
  Denote by $\MDP(G)\subseteq \MD(G)$ the set of strong \emph{prime}
  modules of $G$.
  
  Furthermore, let $\mb{M}_i$ denote a set of modules that have been merged
  w.r.t.\ $\Go$ resulting in the new module $\mathcal{M}_i = \bigcup_{M\in
  \mb{M}_i} M \in \M(\Go)$, where $\mb{M}_i$ is a subset of some
  $\Pmax(G[M'])$ for some $M'\in \MDP(G)$. In other words, each $\mb{M}_i$
  contains only the children of strong prime modules of
  $G$.
  
  Let $\mb{M}_1,\dots,\mb{M}_r$ be 
  all these sets of modules of $G$ that have been merged w.r.t.\ $\Go$ into
  respective modules $\mathcal{M}_1,\dots, \mathcal M_r$ and
  $\mathcal M = \{ \mathcal{M}_1,\dots, \mathcal M_r\}$ denote the set of
  these resulting new modules in $\Go$. Then 
  \[ F_{\opt} = F_{\mathcal M} :=\bigcup_{i=1}^r 
     F_{\Go}\bigg(\Merge_{M\in \mb M_i}M\to \mathcal M_i \bigg) \] 
\label{thm:Edit=Merge}
 \end{theorem}

\begin{proof}	
  The proof follows an iterative process, starting with the input graph
  $G^0 = G$, and proceeds by stepwisely editing \emph{within} certain
  strong prime modules $M$ resulting in new graphs $G^1,G^2,\dots, G^n =
  \Go$ and $F_{\opt} = F_{\mathcal{M}}$ for some integer $n$.  In each
  step, which leads from a non-cograph $G^k$ to $G^{k+1}$, we operate on
  one strong prime module $M\in \MD(G^k)$ of $G^k$ that does not contain
  any other strong prime module.  We write $\Pmaxi{k} := \Pmax(G^k[M])$ for
  the maximal modular partition of $G^k[M]$.  Theorem \ref{thm:all} (T4)
  implies that $G^k[M]/\Pmaxi{k}$ is prime and thus does not contain
  twins. Lemma \ref{lem:quotient-cograph} implies that $\Go[M]/\Pmaxi{k}$
  is a cograph and hence, by Theorem \ref{thm:cograph-characterize}
  contains twins.

  We will show in \emph{PART 1} that each equivalence class $N$ of the twin
  relation $N\eqcl\T$ on $\Go[M]/\Pmaxi{k}$ yields a set of vertices $N_V
  \subseteq V$, where each $N_V$ is a module of $\Go$ but not of $G^k$. In
  particular, every such $N_V$ is obtained by merging modules of $G^k$ only
  so that only $N_V$ is affected.  Application of merge edit operations
  contained in $\Fo$ to obtain these new modules $N_V$ results in the new
  graph $G^{k+1}$. The procedure is then repeated unless the new graph
  $G^{k+1}$ equals $\Go$.

  We then show in \emph{PART 2}, that none of the new modules $N_V$ is a
  module of the starting graph $G$. Furthermore we will see that $N_V$ is
  obtained not only by merging modules of $G^{k}$ but also by merging
  modules of $G$.  Finally we use these two results to show that $F_{\opt}
  = F_{\mathcal M}$.  

\par\bigskip\noindent\emph{PART 1:} \\
We start with the base case, and show how to obtain the graph $G^1$ from
$G^0:=G$. Set $F^0 = \Fo$. W.l.o.g., assume that $G^0$ is a not a cograph
and thus that it contains prime modules. Let $M\in\MD(G^0)$ be a strong
prime module of $G^0$ that does not contain any other strong prime module.

By the preceding arguments, $G^0[M]/\Pmaxi{0}$ is prime and thus does not
contain twins, while $\Go[M]/\Pmaxi{0}$ is a cograph and contains twins. In
particular, there must be vertices (representing strong modules
of $G^0$) in $G^0[M]/\Pmaxi{0}$ that become twins in $\Go[M]/
\Pmaxi{0}$. Let $\T$ be the twin relation of $\Go[M]/\Pmaxi{0}$ and denote,
for a given equivalence classes $N\eqcl \T$, by $N_V \subseteq V$ the set
of twins contained in $N$, i.e., $N_V:=\bigcup_{M_j\in N} M_j =
\bigcup_{M_j\in N}\{v\in M_j\}$.  By Theorem \ref{thm:twins}, each
$N\eqcl\T$ is a module of $\Go[M]/\Pmaxi{0}$.  Lemma
\ref{lem:quotient-cograph} implies that $\Pmaxi{0}$ is a modular partition
of $\Go$ and thus we can apply Theorem \ref{thm:all} (T2) to conclude that
$N_V$ is a module of $\Go[M]$. Since $\Fo$ is module preserving, it follows
that $M$ is a module of $\Go$ and thus, by Lemma \ref{lem:module-subg},
$N_V$ is a module of $\Go$. Corollary \ref{cor:twinClassNotPrime}
furthermore implies that $N_V\neq M$ for all $N\eqcl \T$.  Moreover, by
Theorem \ref{thm:all} (T5), none of the equivalence classes $N\eqcl T$ with
$|N|>1$ are modules of $G^0[M]/\Pmaxi{0}$. Hence, by Theorem \ref{thm:all}
(T2), $N_V$ is not a module of $G^0[M]$ for all $N\eqcl \T$ with $|N|>1$
and hence, Lemma \ref{lem:module-subg} implies that $N_V$ is not a module
of $G^0$ for all $N\eqcl\T$ with $|N|>1$. Since each equivalence class
$N\eqcl \T$ with $|N|>1$ comprises modules of $G^0 = G$, and for the
respective vertex set $N_V$ we have both $N_V\notin \M(G^0)$ and
$N_V\in\M(\Go)$, the modules contained in $N$ are merged to $N_V$ w.r.t.\
$\Go$, in symbols $\Merge_{M_i\in N} M_i \to N_V$ for all $N\eqcl \T$ with
$|N|>1$.

\smallskip

We write $\mathcal{M}^1 := \{N_V \mid N_V =\cup_{M_i\in N} M_i, N\eqcl \T,
|N|>1\}$ for the set of all new modules $N_V$ that are the result of
merging modules of $G^0[M]$. We continue to construct the graph $G^1$.  To
this end, we set $F^0=\Fo$ and show that the subset $E^0\subseteq \Fo$,
where $E^0$ contains all pairs of vertices $\{x,y\}$ with $x\in N_V$ and
$y\in M\setminus N_V$ for all $N_V\in\mathcal{M}^1$, is non-empty. Let
$F(\T) =\{(M_i,M_j)\in F^0(M,\Pmaxi{0})\mid M_i\in N, M_j\notin N, N\eqcl
\T, |N|>1\}$, where $F^0(M,\Pmaxi{0})$ is the edit set implied by $F^0$ as 
defined in Remark \ref{rem:all}.

Theorem \ref{thm:twins} implies that for all $x,y\in N_V$ it holds that
$\{x,y\}\notin \Fo$ for any $N\eqcl \T$. As either all or none of the edges
between any two modules are edited, we conclude that $\{M_i,M_j\}\notin
F(\T)$ holds for all $N\eqcl\T$ and all $M_i,M_j\in N$.  However, since
$G^0[M]/\Pmaxi{0}$ does not contain any pair of twins, while
$\Go[M]/\Pmaxi{0}$ does, it must have $F(\T) \neq \emptyset$.  Hence, $E^0
= \{\{x,y\}\in F^0 \mid x\in M_i, y\in M_j, (M_i,M_j)\in F(\T)\}\neq
\emptyset$. Theorem \ref{thm:twins} implies $\{x,y\}\notin E^0$ for all
$x\in M_i, y\in M_j$, all $M_i,M_j\in N$ and all $N\eqcl \T$. Thus,
$\{x,y\}\notin E^0$ for all $x,y\in N_V$.  Since $E^0\neq \emptyset$, it
contains only pairs of vertices $\{x,z\}$ with $x\in N_V\subseteq M$ and
$z\in M\setminus N_V$. In other words, $E^0$ comprises only pairs of
vertices $\{x,y\}\subset M$ that have been edited to merge the modules of
$G^0[M]$ resulting in the new modules $N_V\in \mathcal M^1$. Note that
there might be additional edits $\{x,y\}\in \Fo\setminus E^0$ for some
$x\in N_V$ and $y\in V\setminus N_V$.  Nevertheless we have $E^0\subseteq
\bigcup_{N\eqcl \T} F_{\Go}(\Merge_{M_i\in N} M_i \to N_V)\subseteq
F_\mathcal{M}$.

\smallskip

Finally, set $G^1 = (V, E\DELTA E^0)$. Thus, all $N_V\in \mathcal M^1$ are
modules of $G^1$.  If $G^1 = \Go$, then $E^0 = \Fo$. By construction $E^0
\subseteq F_{\mathcal M} \subseteq \Fo$. Hence we can conclude that
$F_\mathcal{M} = \Fo$ and we are done. 
	
\smallskip 

Now we turn to the general editing step. If $G^k$, $k\geq1$ is not a
cograph, then we define the set $F^{k}=F^{k-1}\setminus E^{k-1}$.  We can
re-use exactly the same arguments as above. Starting with a strong prime
module $M\in \MD(G^{k})$ (that does not contain any other strong prime
module), we show that the resulting set $E^{k}$ is not empty and comprises
(non)edges in $F^{k}\subseteq \Fo$ so that $E^k\subseteq
F_{\mathcal{M}}$. In particular, we find that $E^k$ contains only the
(non)edges that have been edited so that for all new modules $N_V\in
\mathcal M^{k+1}$ of the graph $G^{k+1} = (V,E\DELTA (\cup_{j=0}^{k} E^j))$
it holds that $N_V\notin \M(G^{k})$, $N_V\in \M(\Go)$ and that $N_V$ is the
union of modules contained in $\M(G^{k})$ that are, in particular, children
of the chosen prime module in $G^{k}$.

\par\bigskip\noindent\emph{PART 2:} \\
It remains to show that every $N_V\in\mathcal M^{k+1}$ has the following two
properties: 
\begin{itemize}
\item[(a)] $N_V\notin \M(G)$; and 
\item[(b)] $N_V$ is a module that has been obtained by merging modules that
  are children of prime modules of $G$ and not only by merging such modules
  of $G^{k}$.
\end{itemize}
We first prove the following
\par\noindent\emph{Claim 1:} 
Any module of $G^0=G$ is also a module
of $G^k$.

\par\noindent\emph{Proof of Claim 1.}
\par\noindent We proceed by induction. 
Let $k=1$ and thus, $G^1 = (V, E\DELTA E^0)$. Let $M^*$ be a module of
$G^0=G$. If $\{x,a\}\notin E^0$ for all vertices $x\in M^*$ then $M^*$ is a
module in $G^1$ because the \out{M^*}-neighborhood in $G^1$ is not
changed. Now assume that there is some vertex $x\in M^*$ with $\{x,a\}\in
E^0$. If all such $a$ with $\{x,a\}\in E^0$ are also contained in $M^*$,
then $M^*$ is still a module in $G^1$. Therefore, in what follows assume
that $a\notin M^*$.

Let $M$ be the  strong prime module that has been used to edit $G^0$ to
$G^1$ in our construction above.  Since for all $\{x,a\}\in E^0$ we have by
construction $x,a\in M$, we can conclude that $M\cap M^* \neq
\emptyset$. Since $M$ is a strong module we have $M\subseteq M^*$ or
$M^*\subseteq M$. However, from $a\in M$ and $a\notin M^*$ it follows that
$M^* \subsetneq M$. Write $\Pmaxi{0}$ for the maximal modular partition of
$G^0[M]$. Theorem \ref{thm:all} (T5) implies that for any non-trivial
subset $\mb P'$ of $\Pmaxi{0}$ the union of elements cannot be a module of
$G^0$. In other words, $M^*$ is not the union of elements of any such
subset $\mb P'$. Therefore, if $M^* \subsetneq M$, then $M^* = M_i$ for
some module $M_i\in \Pmaxi{0}$. If $M_i=\{x\}$, then $M_i$ trivially
remains a module of $G^1$. Hence assume $|M_i|>1$. By construction
$M_i\subseteq N_V$ for some $M_i\in N\eqcl \T$. There are two cases: (1)
$|N|>1$ and thus $M_i\subsetneq N_V$, and (2) $|N|=1$, i.e., $N=\{M_i\}$,
and thus $N_V = M_i$ is not merged.
\begin{itemize}
\item[{Case (1)}] Assume first that the edge $(x,a) \in E^0$ with $x\in
  M_i$ and $a\notin M_i$ was added. Since $\Fo$ is module preserving,
  $N_V\in \M(\Go)$. Hence, by construction of $E^0$, we can conclude that
  all vertices $y\in N_V\subseteq M$ must be adjacent to $a$ in
  $G^1[M]$. Since $M_i\subseteq N_V$, all $y\in M_i$ must be adjacent to
  $a$ in $G^1[M]$.  Analogously, if the edge $(x,a)$ has been removed, then
  all $y\in M_i$ must be non-adjacent to $a$. Hence, elements in $M_i$ have
  the same \out{M_i}-neighborhood in $G^1[M]$ and thus, $M_i$ is a module
  of $G^1[M]$.  Since our construction does not change the
  \out{M}-neighborhood, i.e., $M$ is a module of $G^1$, Lemma
  \ref{lem:module-subg} implies that $M_i$ is also a module of $G^1$.
\item[{Case (2)}] Recall first that $E^0$ comprises only pairs of vertices
  $\{u,v\}\subset M$ that have been edited to merge the modules of $G^0[M]$
  resulting in the new modules in $\mathcal M^1$. Since $N_V$ is not merged
  we have $N_V\notin \mathcal M^1$ and thus, we cannot directly assume that
  $\{y,a\}$ with $y\in N_V=M_i$ are in $E^0$. However, as there was some
  edit $\{x,a\}\in E^0$ there must be some other merged module $N'_V\in
  \mathcal{M}^1$ with $a\in N'_V$. Moreover, since $\Fo$ is module
  preserving it follows that $M_i$ is a module of $\Go$. Hence, if $\{x,a\}
  \in E^0\subseteq \Fo$ was added, then all pairs of vertices $\{y,a\}$
  with $y\in M_i$ not adjacent to $a$ must be contained in $\Fo$ and thus,
  in particular, $a\in N'_V$, which implies $\{y,a\}$ in $E^0$. Analogous
  arguments apply if $\{x,a\}$ was removed. Therefore, either all vertices
  in $M_i$ are adjacent to $a$ or all of them as nonadjacent to $a$ in
  $G^1[M]$. Therefore $M_i$ is a module of $G^1[M]$. As in case (1) we can
  now argue that $M_i$ is also a module of $G^1$.
\end{itemize}
To summarize, all modules of $G^0=G$ are modules of $G^1$. Assume the
statement is true for $k$ and assume that $G^{k}$ is not a cograph (since
there is nothing more to show if $G^{k}$ is a cograph). Applying the same
arguments as above, we can infer that every module of $G^{k}$ is also a
module of $G^{k+1}$. Since all modules of $G$ are by assumption also
modules of $G^{k}$, they are also modules of $G^{k+1}$.  
\hfill{$\circ$}
\smallskip	

\par\noindent\emph{Proof of Statement (a).}
\par\noindent 
By Claim 1., every module of $G$ is a module of $G^{k}$. Thus, if
there is a subset $M\subseteq V$ that is not a module of $G^k$, then $M$ is
not a module of $G$. Since we have already shown that all modules $N_V\in
\mathcal M^{k+1}$ of $G^{k+1}$ are not modules of $G^{k}$, we conclude that
none of the modules $N_V\in \mathcal M^{k+1}$ can be a module of $G$. This
implies statement (a).
\hfill{$\diamond$}

\smallskip
\par\noindent\emph{Proof of Statement (b).}
\par\noindent By Lemma \ref{lem:strongHG} and construction of $G^k$, we
find that all strong prime modules of $G^k$ are also strong prime modules
of $G^{k-1}$. Therefore, by induction, any strong prime module of $G^k$ is
also a strong prime module of $G$.
 
It remains to show that the children of the chosen strong prime module $M$
in $G^k$ are also the children of $M$ in $G$. In other words, we must show
that for the maximal modular partitions we have that $\Pmax(G[M]) =
\Pmax(G^k[M]))$ for each $k\geq 1$.

We again proceed by induction: By construction, if $M$ is the chosen strong
prime module of $G^0$ that does not contain any other strong prime module
in $G^0$, then only children of $M$ in $G^0$ are merged to obtain
$G^1$. Therefore, let $M$ be the chosen strong prime module of $G^1$ that
does not contain any other strong prime module in $G^1$ and that is used to
obtain the graph $G^{2}$.  Lemma \ref{lem:strongHG} implies that $M$ is a
strong prime module of $G$ and
$\Pmax(G[M])=\Pmax(G^0[M])=\Pmax(G^1[M])$. Thus, only children of prime
modules of $G$ have been merged to obtain the graph $G^2$. Assume the
statement is true for $k$. By analogous arguments as in the step from $G^0$
to $G^1$, we can show that for the chosen strong prime module $M$ in $G^k$
that does not contain any other strong prime of $G^k$ to obtain $G^{k+1}$,
we have $\Pmax(G^{k-1}[M_k])=\Pmax(G^k[M])$. Thus, only children of prime
modules of $G^{k-1}$ have been merged to obtain the graph
$G^{k+1}$. However, since each strong prime module of $G^k$ is a strong
prime module of $G$ and since by induction hypothesis $\Pmax(G^{k-1}[M]) =
\Pmax(G[M])$ we obtain that $\Pmax(G^{k}[M]) = \Pmax(G[M])$, and thus only
children of prime modules of $G$ have been merged to obtain the graph
$G^{k+1}$.  \hfill{$\diamond$}
\bigskip

Finally, recall that $E^{k} \subseteq \{\{x,y\}\in \Fo \mid x\in N_V, y\in
V\setminus N_V, N_V\in \mathcal M^{k}\}= \bigcup_{N\eqcl \T}
F_{\Go}(\Merge_{M_i\in N} M_i \to N_V)\subseteq F_\mathcal{M}$, where $\T$
is the twin relation applied to $\Go[M]/\Pmaxi{k}$ of the chosen strong
prime module $M$ in $G^k$ that does not contain any other strong prime
module.  By construction and because any module contained in $\mathcal M' =
\{N_V\in M^j\mid 1\leq j\leq k\}$ is obtained by merging the children of
strong prime modules $G$ only and since the children of strong prime
modules are in particular modules of $G$, we have $\bigcup_{j=1}^{k} E^{j}
\subseteq \bigcup_{j=1}^{k}\bigcup_{N_V\in \mathcal M^j}
F_{\Go}(\Merge_{M\in N} M \to N_V)\subseteq F_\mathcal{M}$, where $N$
comprises all modules that have been merged to obtain the respective module
$N_V$ contained in $\mathcal M'$.

This iterative process may lead to the situation that $\mathcal
M'\subsetneq \mathcal M$, where $\mathcal M$ denotes the set of all modules
of $\Go$ that have been obtained by merging modules of $G$.  However, since
$E^{k} \subseteq \Fo$ and, in particular, $E^{k}\neq \emptyset$ if and only
if $G^{k}$ is a non-cograph (and thus contains prime modules), and because
$|F^{k+1}|<|F^k|$, the iteration necessarily terminates for some finite $n$
with $F^n=\emptyset$ and thus $G^n = (V, E\DELTA (\cup_{j=1}^{n} E^j)) =
\Go$ and $\cup_{i=1}^n E^i = \Fo$.  Since $\cup_{i=1}^n E^i \subseteq
F_{\mathcal{M}'} \subseteq F_{\mathcal{M}}\subseteq \Fo$ we can finally
conclude that $F_{\mathcal{M}} = \Fo$.   
\end{proof}

Theorem \ref{thm:Edit=Merge} implies that every module-preserving optimal
edit set $\Fo$ (which always exist) can be expressed as optimal module
merge edits. Hence, instead of editing the induced $P_4$'s of a given graph
$G$ directly, one can \emph{equivalently} resolve the strong prime modules
by merging their children until no further prime module is left. At the
first glance this result seems to be only of theoretical interest since the
construction of optimal module merge edits is not easier than solving
the editing problem. 

The proof is constructive, however, and implies an alternative \emph{exact}
algorithm to solve the cograph editing problem, this time based on the
stepwise resolution of the prime modules. Of course it has exponential
runtime because in each step one needs to determine which of the modules
have to merged and which of the (non)edges have to edited to obtain new
modules. In particular, there are $2^n$ possible subsets for a prime module
with $n$ children in $\MDT(G)$ that give all rise to modules that can be
merged to new ones. Moreover, for each subset of modules that will be
merged, there are exponentially many possibilities in the number of
vertices to add or remove edges. 

\section{A modular-decomposition-based Heuristic for Cograph Editing}

The practical virtue of our result is that it suggests an alternative
strategy to construct heuristic algorithms for cograph editing. Our
starting point is Lemma \ref{lem:pair-merge}, which implies that we can
construct $\Fo$ by pairwise merging of children of prime modules.

Assume there is strong prime module $M$ and $\mb M$ denotes a subset
children of $M$ w.r.t.\ $MDT(G)$ that have been merged w.r.t.\ $\Go$.  Now,
instead of merging all modules at once, one can perform the merging process
step by step. Take e.g.\ $M_{1}$ and $M_2$ of $\mb M$ and define
$F_{1,2}\subseteq \Fo$ as the set of all edits that have been used to merge
$M_{1}$ and $M_2$ so that they become twins in $\Go[M]$, delete $M_1$ from
$\mb M$ and replace $M_{2}$ by $M_{1}\cup M_2$ in $\mb M$ and $\Fo$ by
$\Fo\setminus F_{1,2}$. Now all vertices in $M_{1}\cup M_2$ have the same
\out{M_{1}\cup M_2}-neighbors in $G[M]\DELTA F_{1,2}$.  Repeating this
procedure reduces the size of $\mb M$ by one element in each round and
eventually terminates with $F(\mb M):=\cup_{i=1}^{l-1} F_{i,i+1}\subseteq
\Fo$.

This strategy can be applied to all sets $\mb{M}^1=\mb{M}, \mb{M}^2, \ldots
\mb{M}^r$ that contain modules that are children of $M$ in $\MDT(G)$ and
that have been merged to some new module $\mc{M}^1=\mc{M}, \mc{M}^2, \ldots
\mc{M}^r$, respectively.  Hence, $\cup_{i=1}^r F(\mb M^i) \subseteq
\Fo[M]$, where $\Fo[M]$ contains all (non-)edges $\{x,y\}\in \Fo$ with $x,y\in
M$.  Lemma \ref{lem:opt-modul-edit} implies that $\Fo[M]$ is optimal in
$G[M]$ and hence, the modules in $\mb{M}^1, \mb{M}^2, \ldots \mb{M}^r$
cannot be merged with fewer edits than in $\Fo[M]$. It follows that
$\cup_{i=1}^r F(\mb{M}^i) = \Fo[M]$.  Thus, there is always an optimal edit
set $\Fo$ that can be expressed by means of successively \emph{pairwise}
merge module edit, as done above. Algorithm \ref{alg:coH} summarizes this
approach in the form of pseudocode. 

\begin{algorithm}[htbp]
  \caption{Simple Cograph Editing Heuristic.\newline \emph{ Two functions,
    \texttt{get-module()} and \texttt{get-module-pair-edit()}, influence the
    practical efficiency, see text for details of their specification.}  }
\label{alg:coH}
\begin{algorithmic}[1]
  \STATE \textbf{INPUT:} A graph $G=(V,E)$;
  \STATE Compute $\MD(G)$
  \WHILE{$M =$\texttt{get-module(}$\MD(G)$\texttt{)}}
  \IF{$G[M]$ is a spider}
  \STATE Optimally edit $G[M]$ to a cograph by application of  
  \texttt{EDP4} \cite{Liu:11,Liu:12}
  \ELSE
  \STATE Let $M_1,\ldots,M_K$ be the children of $M$ in $\MDT(G)$, 
  i.e., $\{M_1,\ldots,M_K\} = \Pmax(G[M])$.
  \STATE $(M_i,M_j, F_{i,j})=$ 
  \texttt{get-module-pair-edit(}$\{M_1,\ldots,M_K\}$\texttt{)} 
  \STATE $G\gets (V,E(G)\DELTA F_{i,j})$
  \ENDIF
  \ENDWHILE		
  \STATE \textbf{OUTPUT:} The cograph $G^*$;
\end{algorithmic}
\end{algorithm}

Algorithm \ref{alg:coH} contains two points at which the choice a
particular module or a particular pair of modules affects performance and
efficiency. First, the function \texttt{get-module()} returns a strong
prime module that does not contain any other prime module and returns
\texttt{false} if there is no such module, i.e., if $G$ is a cograph.
Second, subroutine \texttt{get-module-pair-edit()} extracts from
$\{M_1,\ldots,M_K\}$ a pair of modules. Ideally, these should satisfy the
the following two conditions:
\begin{itemize}
\item[{(i)}] $M_i$ and $M_j$ have a minimum number of edits so that the
  \out{M_i\cup M_j}-neighborhood in $G[M]$ becomes identical for
  all $x,y\in {M_i\cup M_j}$ among all pairs  in $\Pmax(G[M])$, and 
\item[{(ii)}] additionally maximizes the number of removed $P_4$'s in
  $G[M]$ after applying these edits.
\end{itemize}

\begin{lemma}
  If \texttt{get-module-pair-edit()} is an ``oracle'' that always returns a
  correct pairs $M_i$ and $M_j$ together with the respective edit set $F_{i,j}$
  that is used to merged them and
  \texttt{get-module()} returns and arbitrary strong prime module that does
  not contain any other prime module, then Alg.\ \ref{alg:coH} computes an
  optimally edited cograph $\Go$ in $O(p\Lambda h(n))\le O(n^2 h(n))$ time,
  where $p$ denotes the number of strong prime modules,
  $\Lambda=\max|\Pmax(M)|$ among all strong prime modules of $G$, and 
  $h(n)$ is the cost for evaluating \texttt{get-module-pair-edit()}.
\end{lemma}

\begin{proof}
  The correctness of Algorithm\ref{alg:coH} follows directly from Lemma
  \ref{lem:pair-merge} and Theorem \ref{thm:Edit=Merge}.
  
  The modular decomposition $\MD(G)$ can be computed in linear-time, 
  see \cite{CH:94,DGC:01,CS94,CS:99, TCHP:08}.
  Then, we have to resolve each of the $p$ modules and in each step in the worst
  case all modules have to be merged stepwisely, resulting an effort of
  $O(|\Pmax|)$ merging steps in each iteration. 
  Since $p\leq n$ and $\Lambda\leq n$ we obtain $O(n^2 h(n))$ as an upper
  bound.  
\end{proof}

In practice, the exact computation of the optimal editing pairs requires
exponential effort. Practical heuristics for
\texttt{get-module-pair-edit()}, however, can be implemented in polynomial
time. A simple heuristic strategy to find those pairs can be established as
follows: Mark all of the $O(\Lambda^2)$ pairs $(M_i,M_j)$ where the set
$\Gamma=N^{M_i}_G\DELTA N^{M_j}_G\setminus (M_i\cup M_j)$ of distinct
\out{M_i}- and \out{M_j}-neighbors that are not contained in $M_i$ and $M_j$
has minimum cardinality. Removing, resp., adding all edges $xy$ with $x\in
M_i\cup M_j$, $y\in \Gamma$ would yield a new module $M_i\cup M_j$ in the
updated graph. Among all those marked pairs take the pair for a final merge
that additionally removes a maximum number of induced $P_4$'s in the course
of adjusting the respective out-neighborhoods. This amounts to an efficient
method for detecting induced $P_4$'s.  A detailed numerical evaluation of
 heuristics for cograph editing will be discussed elsewhere.

\bibliographystyle{alpha}
\bibliography{biblio}

\newcommand{\etalchar}[1]{$^{#1}$}
\begin{thebibliography}{HHRH{\etalchar{+}}13}

\bibitem[BCF{\etalchar{+}}15]{BCF+15}
S.~J. Berkemer, R.~R.~C. Chaves, A.~Fritz, M.~Hellmuth, M.~Hernandez-Rosales,
  and P.~F. Stadler.
\newblock Spiders can be recognized by counting their legs.
\newblock {\em Math. Comp. Sci.}, pages 1--5, 2015.

\bibitem[BCHP08]{BCHP:08}
A.~Bretscher, D.~Corneil, M.~Habib, and C.~Paul.
\newblock A simple linear time lexbfs cograph recognition algorithm.
\newblock {\em SIAM J. on Discrete Mathematics}, 22(4):1277--1296, 2008.

\bibitem[Bla78]{Blass:78}
Andreas Blass.
\newblock Graphs with unique maximal clumpings.
\newblock {\em Journal of Graph Theory}, 2(1):19--24, 1978.

\bibitem[BLS99]{BLS:99}
Andreas Brandst\"{a}dt, Van~Bang Le, and Jeremy~P. Spinrad.
\newblock {\em Graph Classes: A Survey}.
\newblock Society for Industrial and Applied Mathematics, Philadelphia, PA,
  USA, 1999.

\bibitem[Cai96]{Cai:96}
Leizhen Cai.
\newblock Fixed-parameter tractability of graph modification problems for
  hereditary properties.
\newblock {\em Information Processing Letters}, 58(4):171 -- 176, 1996.

\bibitem[CH94]{CH:94}
Alain Cournier and Michel Habib.
\newblock A new linear algorithm for modular decomposition.
\newblock In Sophie Tison, editor, {\em Trees in Algebra and Programming -
  CAAP'94}, volume 787 of {\em Lecture Notes in Computer Science}, pages
  68--84. Springer Berlin Heidelberg, 1994.

\bibitem[CJS72]{CJS:72}
D.D. Cowan, L.O. James, and R.G. Stanton.
\newblock Graph decomposition for undirected graphs.
\newblock In {\em 3rd S-E Conference on Combinatorics, Graph Theory and
  Computing}, Utilitas Math, pages 281--290, 1972.

\bibitem[CLSB81]{Corneil:81}
D.~G. Corneil, H.~Lerchs, and L.~Steward~Burlingham.
\newblock Complement reducible graphs.
\newblock {\em Discr. Appl. Math.}, 3:163--174, 1981.

\bibitem[CPS85]{Corneil:85}
D.~G. Corneil, Y.~Perl, and L.~K. Stewart.
\newblock A linear recognition algorithm for cographs.
\newblock {\em SIAM J. Computing}, 14:926--934, 1985.

\bibitem[DGM97]{DGC:1997}
Elias Dahlhaus, Jens Gustedt, and Ross~M. McConnell.
\newblock Efficient and practical modular decomposition.
\newblock In {\em Proceedings of the Eighth Annual ACM-SIAM Symposium on
  Discrete Algorithms}, SODA '97, pages 26--35, Philadelphia, PA, USA, 1997.
  Society for Industrial and Applied Mathematics.

\bibitem[DGM01]{DGC:01}
Elias Dahlhaus, Jens Gustedt, and Ross~M McConnell.
\newblock Efficient and practical algorithms for sequential modular
  decomposition.
\newblock {\em Journal of Algorithms}, 41(2):360 -- 387, 2001.

\bibitem[EGMS94]{EGMS:94}
A.~Ehrenfeucht, H.N. Gabow, R.M. Mcconnell, and S.J. Sullivan.
\newblock An o(n2) divide-and-conquer algorithm for the prime tree
  decomposition of two-structures and modular decomposition of graphs.
\newblock {\em Journal of Algorithms}, 16(2):283 -- 294, 1994.

\bibitem[Gal67]{gallai-67}
T.~Gallai.
\newblock Transitiv orientierbare graphen.
\newblock {\em Acta Mathematica Academiae Scientiarum Hungarica},
  18(1-2):25--66, 1967.

\bibitem[GHN13]{Gao:13}
Yong Gao, Donovan~R. Hare, and James Nastos.
\newblock The cluster deletion problem for cographs.
\newblock {\em Discrete Mathematics}, 313(23):2763 -- 2771, 2013.

\bibitem[GPP10]{GPP:10}
Sylvain Guillemot, Christophe Paul, and Anthony Perez.
\newblock On the (non-) existence of polynomial kernels for pl-free edge
  modification problems.
\newblock In {\em Parameterized and Exact Computation}, pages 147--157.
  Springer, 2010.

\bibitem[HDMP04]{habib2004simple}
Michel Habib, Fabien De~Montgolfier, and Christophe Paul.
\newblock A simple linear-time modular decomposition algorithm for graphs,
  using order extension.
\newblock In {\em Algorithm Theory-SWAT 2004}, pages 187--198. Springer, 2004.

\bibitem[HHRH{\etalchar{+}}13]{HHH+13}
M.~Hellmuth, M.~Hernandez-Rosales, K.~T. Huber, V.~Moulton, P.~F. Stadler, and
  N.~Wieseke.
\newblock Orthology relations, symbolic ultrametrics, and cographs.
\newblock {\em Journal of Mathematical Biology}, 66(1-2):399--420, 2013.

\bibitem[HM79]{HM-79}
M.~Habib and M.C. Maurer.
\newblock On the {X}-join decomposition for undirected graphs.
\newblock {\em Discrete Appl. Math.}, 3:198--207, 1979.

\bibitem[Ho{\`a}85]{Hoang:85}
C.~T. Ho{\`a}ng.
\newblock {\em Perfect graphs}.
\newblock PhD thesis, School of Computer Science, McGill University, Montreal,
  Canada, 1985.

\bibitem[HP10]{HP:10}
M.~Habib and C.~Paul.
\newblock A survey of the algorithmic aspects of modular decomposition.
\newblock {\em Computer Science Review}, 4(1):41 -- 59, 2010.

\bibitem[HPV99]{HPV:00}
Michel Habib, Christophe Paul, and Laurent Viennot.
\newblock Partition refinement techniques: An interesting algorithmic tool kit.
\newblock {\em International Journal of Foundations of Computer Science},
  10(02):147--170, 1999.

\bibitem[HWL{\etalchar{+}}15]{HWL+15}
Marc Hellmuth, Nicolas Wieseke, Marcus Lechner, Hans-Peter Lenhof, Martin
  Middendorf, and Peter~F. Stadler.
\newblock Phylogenomics with paralogs.
\newblock {\em Proceedings of the National Academy of Sciences},
  112(7):2058--2063, 2015.

\bibitem[JO89]{JO:89}
B.~Jamison and S.~Olariu.
\newblock P4-reducible-graphs, a class of uniquely tree representable graphs.
\newblock {\em Stud. Appl. Math.}, 81:79--87, 1989.

\bibitem[JO92]{Jamison:92}
B~Jamison and S~Olariu.
\newblock Recognizing {P}$_4$-sparse graphs in linear time.
\newblock {\em SIAM J. Comput.}, 21:381--406, 1992.

\bibitem[LWGC11]{Liu:11}
Yunlong Liu, Jianxin Wang, Jiong Guo, and Jianer Chen.
\newblock Cograph editing: Complexity and parametrized algorithms.
\newblock In B.~Fu and D.~Z. Du, editors, {\em COCOON 2011}, volume 6842 of
  {\em Lect. Notes Comp. Sci.}, pages 110--121, Berlin, Heidelberg, 2011.
  Springer-Verlag.

\bibitem[LWGC12]{Liu:12}
Yunlong Liu, Jianxin Wang, Jiong Guo, and Jianer Chen.
\newblock Complexity and parameterized algorithms for cograph editing.
\newblock {\em Theoretical Computer Science}, 461(0):45 -- 54, 2012.

\bibitem[M{\"o}h85]{Moh:85}
R.H. M{\"o}hring.
\newblock Algorithmic aspects of the substitution decomposition in optimization
  over relations, set systems and boolean functions.
\newblock {\em Annals of Operations Research}, 4(1):195--225, 1985.

\bibitem[MR84]{MR-84}
R.H. M{\"o}hring and F.J. Radermacher.
\newblock Substitution decomposition for discrete structures and connections
  with combinatorial optimization.
\newblock In R.A. Cuninghame-Green R.E.~Burkard and U.~Zimmermann, editors,
  {\em Algebraic and Combinatorial Methods in Operations Research Proceedings
  of the Workshop on Algebraic Structures in Operations Research}, volume~95 of
  {\em North-Holland Mathematics Studies}, pages 257 -- 355. North-Holland,
  1984.

\bibitem[MS89]{MS:89}
John~H. M\"{u}ller and Jeremy Spinrad.
\newblock Incremental modular decomposition.
\newblock {\em J. ACM}, 36(1):1--19, January 1989.

\bibitem[MS94]{CS94}
Ross~M. McConnell and Jeremy~P. Spinrad.
\newblock Linear-time modular decomposition and efficient transitive
  orientation of comparability graphs.
\newblock In {\em Proceedings of the Fifth Annual ACM-SIAM Symposium on
  Discrete Algorithms}, SODA '94, pages 536--545, Philadelphia, PA, USA, 1994.
  Society for Industrial and Applied Mathematics.

\bibitem[MS99]{CS:99}
Ross~M. McConnell and Jeremy~P. Spinrad.
\newblock Modular decomposition and transitive orientation.
\newblock {\em Discrete Mathematics}, 201(1-3):189 -- 241, 1999.

\bibitem[MS00]{CS:00}
Ross~M. Mcconnell and Jeremy~P. Spinrad.
\newblock {Ordered Vertex Partitioning}.
\newblock {\em Discrete Mathematics and Theoretical Computer Science},
  4(1):45--60, 2000.

\bibitem[NG12]{Nastos:12}
James Nastos and Yong Gao.
\newblock Bounded search tree algorithms for parameterized cograph deletion:
  Efficient branching rules by exploiting structures of special graph classes.
\newblock {\em Discr. Math. Algor. Appl.}, 4:1250008, 2012.

\bibitem[PDdSS09]{Protti:09}
F{\'a}bio Protti, Maise Dantas~da Silva, and Jayme~Luiz Szwarcfiter.
\newblock Applying modular decomposition to parameterized cluster editing
  problems.
\newblock {\em Th. Computing Syst.}, 44:91--104, 2009.

\bibitem[Sei74]{Seinsche1974}
D~Seinsche.
\newblock On a property of the class of n-colorable graphs.
\newblock {\em Journal of Combinatorial Theory, Series B}, 16(2):191 -- 193,
  1974.

\bibitem[SS03]{sem-ste-03a}
Charles Semple and Mike Steel.
\newblock {\em Phylogenetics}, volume~24 of {\em Oxford Lecture Series in
  Mathematics and its Applications}.
\newblock Oxford University Press, Oxford, UK, 2003.

\bibitem[TCHP08]{TCHP:08}
Marc Tedder, Derek Corneil, Michel Habib, and Christophe Paul.
\newblock Simpler linear-time modular decomposition via recursive factorizing
  permutations.
\newblock In {\em Automata, Languages and Programming}, volume 5125 of {\em
  Lecture Notes in Computer Science}, pages 634--645. Springer Berlin
  Heidelberg, 2008.

\end{thebibliography}
\label{sec:biblio}

\end{document}